\documentclass[11pt]{article}
\usepackage[margin=1in]{geometry}
\usepackage{booktabs}
\usepackage{nicefrac}
\usepackage{bm}
\usepackage{multirow}
\usepackage{authblk}
\usepackage{amsmath,amssymb,amsthm}

\usepackage{xcolor}
\definecolor{ForestGreen}{rgb}{0.1333,0.5451,0.1333}
\definecolor{DarkRed}{rgb}{0.65,0,0}
\definecolor{Red}{rgb}{1,0,0}

\usepackage[linktocpage=true,backref=page,
pagebackref=true,colorlinks,
linkcolor=DarkRed,citecolor=ForestGreen,urlcolor=ForestGreen]
{hyperref}

\usepackage{xspace}
\newcommand{\ranking}{\textsc{ranking}\xspace}
\newcommand{\balance}{\textsc{balance}\xspace}
\newcommand{\wl}{\textsc{water level}\xspace}
\newcommand{\wf}{\textsc{water-filling}\xspace}
\newcommand{\balancedranking}{\textsc{balanced-ranking}\xspace}
\newcommand{\eagerwf}{\textsc{eager water-filling}\xspace}
\newcommand{\reranking}{\textsc{periodic reranking}\xspace}

\usepackage{todonotes}

\newcommand{\calD}{\mathcal{D}}
\newcommand{\E}{\mathbb{E}}

\usepackage{cleveref}

\newtheorem{theorem}{Theorem}[section]

\newtheorem{lemma}[theorem]{Lemma}

\title{Online Matching: A Brief Survey}

\date{\vspace{-1cm}}

\author[1]{Zhiyi Huang}
\author[2]{Zhihao Gavin Tang}
\author[3]{David Wajc}
\affil[1]{University of Hong Kong}
\affil[2]{Shanghai University of Finance and Economics}
\affil[3]{Technion --- Israel Institute of Technology}

\begin{document}

\maketitle

\begin{abstract}
	Matching, capturing allocation of items to unit-demand buyers, or tasks to workers, or pairs of collaborators, is a central problem in economics. 
	Indeed, the growing prevalence of matching-based markets, many of which online in nature, has motivated much research in economics, operations research, computer science, and their intersection.
	This brief survey is meant as an introduction to the area of online matching, with an emphasis on recent trends, both technical and conceptual.
\end{abstract}

\section{Introduction}
Matching theory lies at the heart of Economics, Computation and their intersection. 
Matching markets have played increasingly dominant roles in the world economy, both on the micro and macro scale. 
Such markets arise in domains as varied as Internet advertising, crowdsourcing of work and transportation, and organ transplantation.
The repeated interactions and lack of certainty about future participants (buyers, sellers, etc.)~result in these industries' dynamics being
prime examples of \emph{online} matching markets. 
See \cite{echenique2023online} for detailed discussions.
In this brief survey, we focus on three main aspects of recent developments in the study of such online matching and allocation problems.

Particularly prevalent are online \emph{bipartite} matching markets. Often, agents on one side of the market arrive up front, while agents on the other side are revealed sequentially, to be matched (or not) immediately and irrevocably.
For example, this dynamic abstracts the Internet advertising marketplace, with advertisers known up front and user queries (ad slots) revealed online. This motivates the study of online bipartite matching \cite{karp1990optimal}, and its generalization to weighted settings \cite{aggarwal2011online,feldman2009online} and the AdWords problems \cite{mehta2007adwords}. 
For more on the motivation from the Internet advertising application, see the influential survey on online bipartite matching and ad allocation by \cite{mehta2013online}, and the more recent survey by \cite{devanur2022online}.
We outline recent developments on these online bipartite matching and online ad allocation problems, in \S\ref{sec:central}.

There are, of course, many aspects of modern online matching markets that are not bipartite, or that allow for agents on either side of bipartite matching markets to enter in an interleaved order. 
Similarly, while classic online matching models consider agents as having left the market after matching, in many crowdsourcing marketplaces (e.g., DoorDash, TaskRabbit, Uber/Lyft, etc.) freelance workers return to the market after being assigned a task (i.e., being matched) and completing their tasks.
We discuss recent progress on modeling and addressing such problems, in \S\ref{sec:beyond-ads}.

The above-mentioned sections focus on the robust, but somewhat pessimistic, modeling choice of adversarial inputs and arrival orders. 
A less pessimistic model is that of random-order arrivals (``secretary models''), where the input is generated adversarially, but permuted by nature.
Another modeling choice, motivated by the abundance of historical data from which to learn trends, is to posit a stochastic arrival model with parameters known to the algorithm. 
Here one can compare with either the best offline algorithm (computed by a ``prophet'' who knows the future) or the best online algorithm (computed by a ``philosopher'' who has enough time to think/compute).
We discuss progress on online matching for such stochastic models, and their connection to mechanism design, in \S\ref{sec:prophets}.

Finally, we give a brief glimpse of some overarching techniques that have played key recurring roles in the aforementioned recent developments, in \S\ref{sec:techniques}.
We illustrate some of the ideas with particularly short (and in our opinion, quite teachable) examples of some of the basic techniques in this space, in \S\ref{sec:lesson}.

\section{Online matching and ad allocation}\label{sec:central}

Researchers have made much progress in the past decade on the aforementioned online bipartite matching and ad allocation problems.

The most general problem along this line is \emph{online submodular welfare maximization}.
Consider a set of offline agents (advertisers), and a set of online items (impressions) that arrive one at a time.
Each agent $a$ has a submodular value function $v_a$ over subsets of items;
the algorithm can evaluate $v_a(S)$ for any subset $S$ of the arrived items.
On arrival of an online item, the algorithm must allocate it to an agent immediately and irrevocably.
The basic benchmark is the greedy algorithm that allocates each impression to maximize the immediate increase in social welfare, which is $\nicefrac12$-competitive.
For the general problem, this ratio is optimal for polynomial-time algorithms under standard complexity-theoretic assumptions \cite{kapralov2013online}. Most research has therefore focused on special cases of interest of this problem.

The \emph{(unweighted) online bipartite matching} problem is the special case when each agent either likes or dislikes an item, and is willing to pay $\$1$ to get any one item they like: 
the value $v_a(S)$ equals $1$ if agent $a$ likes at least one item in $S$, and is $0$ otherwise.
Further, when a subset of items $S$ is allocated to agent $a$, we can interpret this as matching $a$ to any one item that $a$ likes in $S$ (e.g., the first one).

In online advertising, some advertisers may be able to pay more than others for an impression that they like.
This can be captured by the \emph{vertex-weighted} generalization of online bipartite matching, where the value $v_a(S)$ is agent $a$'s weight $w_a$ if agent $a$ likes at least one item in $S$, and is $0$ otherwise.

More generally, the same advertiser may have different values for different impressions, e.g., depending on the users' cookies and other information.
This motivates considering edge weights instead of vertex weights; 
the advertiser only pays for one impression like in the unweighted and vertex-weighted case.\footnote{Higher capacity can be simulated by creating multiple offline vertices per advertiser.}
This is the \emph{display ads} problem introduced by \cite{feldman2009online2}, where $v_a(S)$ equals the largest edge weight $w_{ai}$ among items $i \in S$. 

Last but not least, some platforms let advertisers set a daily budget rather than a limit on the number of impressions.
Given the allocation of impressions in a day,  each advertiser $a$ pays either the sum of its values for impressions it gets or its daily budget $B_a$, whichever is smaller;
in other words, $v_a(S) = \min \big\{ \sum_{i \in S} w_{ai}, B_a \big\}$.
Here, the weight $w_{ai}$ is often referred to as agent $a$'s \emph{bid} for impression $i$.
This is the \emph{AdWords} problem introduced by \cite{mehta2007adwords}.

Given the uncertainty over future items, when we decide how to allocate an online item, we do not want to put all our eggs in the same basket.
It is easier to implement this old wisdom when the item is divisible (alternatively, if each agent has a large basket that can take many items).
Imagine that each online item carries one litre of water (a divisible egg);
each offline agent has a bucket (basket) of capacity one litre.
The algorithm distributes an online item's water to its neighbors, where the amount of water going to each neighbor represents the fraction of the item allocated to the agent.
The \balance algorithm (a.k.a., \wl or \wf \cite{azar2006maximizing}) lets the water flow to the least loaded bucket (the basket with the least amount of eggs);
if there are multiple least loaded buckets, the water flows to them at an equal rate.
This algorithm and its generalizations achieve the optimal $1-\nicefrac{1}{e}$ competitive ratio in all the mentioned special cases of online submodular welfare maximization, including unweighted matching~\cite{kalyanasundaram2000optimal}, vertex-weighted matching~\cite{buchbinder2007online},
display ads~\cite{feldman2009online2}, and AdWords~\cite{mehta2007adwords}.

In the original problems where items are indivisible, we can distribute the risk through randomized decisions.
However, making independent random decisions in each round is insufficient for getting a competitive ratio better than $\nicefrac12$, the baseline set by the greedy algorithm \cite{karp1990optimal}.
In the same paper that \cite{karp1990optimal} introduced the online bipartite matching problem, they also gave an elegant \ranking algorithm achieving the optimal $1-\nicefrac{1}{e}$ competitive ratio, later generalized by \cite{aggarwal2011online} to vertex-weighted problem.
The algorithm can be viewed as letting each offline vertex independently set a random price, and having each online vertex choose the lowest price offered by its unmatched neighbors.
See \S\ref{sec:techniques} for a further discussion on this economic interpretation of \ranking and other online matching algorithms.

\subsection{Breaking the $\nicefrac{1}{2}$ Barrier in Longstanding Open Problems}

Recall that for online (vertex-weighted) bipartite matching \ranking achieves an optimal competitive ratio, and in particular breaks the barrier of $\nicefrac{1}{2}$. For display ads and AdWords, however, finding an online algorithm strictly better than the $\nicefrac12$-competitive greedy algorithm had remained elusive for more than a decade.
Fundamentally new ideas seemed 
 necessary because a critical invariant in the analysis of \ranking fails to hold in these two problems.

In 2020, the $\nicefrac12$ barrier was broken for both problems using a new technique called Online Correlated Selection (OCS); see \S\ref{sec:ocs} for a further discussion on this technique.
\cite{fahrbach2020edge} introduced the concept of OCS and gave a $0.508$-competitive algorithm for display ads.
\cite{huang2020adwords} modified the definition of OCS and applied it to the AdWords problem, and as a result, obtained a $0.501$-competitive algorithm.
The OCS technique has then been improved in a series of works by \cite{shin2021making}, \cite{gao2021improved}, and \cite{blanc2021multiway}.
The state-of-the-art competitive ratio for display ads is $0.536$, given by a multi-way OCS algorithm by \cite{blanc2021multiway}.

Despite the aforementioned progress, we remark that there is no known evidence that the $(1-\nicefrac{1}{e})$-competitive ratio cannot be achieved in display ads and AdWords.
Hence, closing the gaps between the upper and lower bounds for these two problems remains an important open problem.

For the general online submodular welfare maximization problem, 
we recall that the simple greedy algorithm is $\nicefrac12$ competitive, and (barring surprising developments in complexity theory) no polynomial-time online algorithms can do better \cite{kapralov2013online}.
That being said, this impossibility relied on the computational hardness of maximizing a submodular function.
It would be interesting to explore online algorithms with unlimited computational capacity, because practical heuristics can often solve these optimization problems better than the worst-case approximation ratio promises, and positive results along this line may point to other special submodular functions that are computationally tractable.

\begin{table}[ht]
	\centering
	\begin{tabular}{lcc}
		\toprule
		& Fractional/Divisible Relaxation & Original Problem \\
		\midrule
		Unweighted & $1-\nicefrac{1}{e}$~\cite{kalyanasundaram2000optimal} & $1-\nicefrac{1}{e}$~\cite{karp1990optimal} \\
		Vertex-Weighted & $1-\nicefrac{1}{e}$~\cite[Section 5]{buchbinder2007online} & $1-\nicefrac{1}{e}$~\cite{aggarwal2011online} \\
		Display Ads & $1-\nicefrac{1}{e}$~\cite{feldman2009online2} & $\bm{0.536}$~\cite{blanc2021multiway} \\
		AdWords & $1-\nicefrac{1}{e}$~\cite{mehta2007adwords} & $\bm{0.501}$~\cite{huang2020adwords} \\
		\bottomrule
	\end{tabular}	
	\caption{State-of-the-art for central online bipartite matching \& allocation problems.}
	\label{table:central}
\end{table}

\subsection{Stochastic Rewards and Oblivious Budget}

Many online advertising platforms adopt the \emph{pay-per-click} model.
In this model, an advertiser pays each time a user clicks on its advertisement.
Since these platforms cannot control the user's behavior, they resort to the next best option:
modeling a user's behavior stochastically, and estimating the probability that the user clicks the advertisement, known as the \emph{click-through-rate} (CTR).
As a result of the users' stochastic behavior, the platform's revenue from assigning an impression to an advertiser is also stochastic.

\cite{mehta2012online} introduced online matching with stochastic rewards.
They analyzed both the \ranking algorithm and a variant of the \balance algorithm, and showed competitive ratios better than $\nicefrac12$ for uniform CTRs.
That is, the CTR is either $p$ or $0$ (if the advertiser is not interested in this impression).

Progress made in the past decade on stochastic rewards is threefold.
First, \cite{mehta2014online} gave the first algorithm that breaks the $\nicefrac12$ barrier for non-uniform but sufficiently small CTRs.
On one hand, small CTRs are arguably the most relevant case in practice because most keywords' CTRs are less than $10\%$. 
On the other hand, $10\%$ or even $1\%$ is larger than the assumption made by \cite{mehta2014online} and its follow-up works.
Hence, it remains an important open problem to design better online algorithms for less restrictive CTRs.

The second line of improvements comes from combining the online primal-dual framework and a more expressive linear program for the problem.
This new analysis method gives a better understanding of classical algorithms \balance~\cite{huang2020online} and \ranking~\cite{huang2023online} in the presence of stochastic rewards.

Researchers have also tried to gain new insight by considering a weaker clairvoyant benchmark.
\cite{goyal2023online} showed that against the weaker benchmark, \ranking achieves the optimal $1-\nicefrac{1}{e}$ competitive ratio for uniform CTRs.
They also analyzed \balance for small CTRs, and obtained a ratio better than the aforementioned state-of-the-art against the offline optimum benchmark.
The latter result was later improved by \cite{huang2023online} to $0.611$.

Last but not least, stochastic rewards are closely related to budget-oblivious algorithms for AdWords, i.e., algorithms which do not know an agent's budget until the moment it is depleted.
By a reduction by \cite{mehta2014online}, a competitive online algorithm for the latter model would yield the same competitive ratio in the former model (but not vice versa).
Again, the greedy algorithm is a $\nicefrac{1}{2}$-competitive budget-oblivious algorithm for this problem.
The survey by \cite{mehta2013online} listed finding a better budget-oblivious algorithm as an open problem.
\cite{vazirani2023towards} suggested a variant of \ranking as a candidate algorithm. \cite{liang2023perturbation} showed that no variant of this algorithm is $(1-\nicefrac{1}{e})$-competitive.
Finally, \cite{udwani2023adwords} proved that the candidate is at least $0.508$-competitive, and a variant of this algorithm is $0.522$-competitive, both under the small-bids assumption, whereby the bids $w_{ai}$'s are small compared to the agent's budget $B_a$.

\section{Beyond Online Bipartite Matching and Ad Allocation}\label{sec:beyond-ads}

The preceding online matching models, largely motivated by online advertising, crucially rely on the assumption that one side of the bipartite graph is fixed and known upfront. This prevents the theory of online matching being applied to other modern applications, including ride-hailing, ride-sharing, rental services, etc.

In this section, we discuss generalizations of classic online bipartite matching. The first two
 generalizations are motivated by ride-hailing and ride-sharing, that allow all vertices to arrive online and allow general (non-bipartite) graphs. The third is somewhat theoretical in nature, but is the most general problem. The last generalization is motivated by rental services and freelance labor markets, and so captures the reusability of resources.

\begin{table}[ht]
	\centering
	\begin{tabular}{lcc}
		\toprule
		& Fractional Relaxation & Original Problem \\
		\midrule
		Fully Online & $0.6$~\cite{tang2022improved} & $0.569$~\cite{huang2020fully2} \\
		General Vertex Arrival & $0.526$~\cite{wang2015two} & $\nicefrac{1}{2}+\Omega(1)$~\cite{gamlath2019online} \\
		Edge Arrival & {$\nicefrac{1}{2}$~\cite{gamlath2019online}} & {$\nicefrac{1}{2}$~\cite{gamlath2019online}}\\
		\multirow{2}{*}{Reusable Resources} & $1-\nicefrac{1}{e}$~\cite{goyal2021reusable} & \multirow{2}{*}{$0.589$~\cite{delong2023reusable}} \\
		& $1-\nicefrac{1}{e}$~\cite{feng2021reusable} & \\
		\bottomrule
	\end{tabular}	
	\caption{State-of-the-art for online matching problems beyond bipartite matching \& ad allocation.}
	\label{table:beyond-ads}
\end{table}

\subsection{Fully Online Model: Vertices with Arrivals and Deadlines}\label{sec:fully} 
In online ride-hailing platforms (e.g., Uber, Lyft, DiDi), ride requests are submitted to the platform in an online fashion and are active in the system for a few minutes. The platform assigns each request to a currently available taxi (or self-employed driver). Requests and taxis can be modeled as vertices in a bipartite graph with edges between compatible ride requests and taxis. 
This is an online bipartite matching problem but does not fit into the classic model, since all vertices (both the requests and the taxis) arrive online. Similarly, ride-sharing platforms, which match ride requests (pairing up riders) are naturally modeled as an online matching problem on \emph{general (non-bipartite) graphs}.

\cite{huang2020fully} introduced the fully online matching model to capture the above scenarios, though the same model was studied earlier by \cite{blum2006online} in the context of liquidity in clearing markets. Let $G=(V,E)$ be the underlying graph, initially \emph{completely unknown}. 
Each time step is either an arrival or a deadline of a vertex. Upon the arrival of a vertex, its incident edges to their previously-arrived neighbors are revealed. 
A vertex can be matched at any point until its \emph{deadline}, with this time revealed on its arrival.
Naturally, we assume the deadline of a vertex is after its arrival, and all edges incident to a vertex are revealed before its deadline.
This model generalizes the classic one-sided online bipartite matching model, where all offline vertices arrive at the beginning and have deadlines at the end, and every online vertex has its deadline right after its arrival.

For the fully online matching problem, \cite{huang2020fully} proved that \ranking achieves a tight $\Omega \approx 0.567$ (the unique solution to $\Omega \cdot e^{\Omega} = 1$) competitive ratio for bipartite graphs, and a competitive ratio of $0.521$ for general graphs.
For the fractional variant of the problem, \cite{huang2019tight} established a tight $2-\sqrt{2} \approx 0.585$ competitive ratio of \balance.
Remarkably, \ranking and \balance are known to be optimal in the classic model, but the claimed tightness here only applies to the two algorithms themselves. 
Indeed, \cite{huang2020fully2} introduced the \balancedranking algorithm that achieves a competitive ratio of $0.569$ for bipartite graphs, and the \eagerwf algorithm that achieves a competitive ratio  of $0.592$ for the fractional variant. The later result was further improved by \cite{tang2022improved} to $0.6$. On the negative side, the state-of-the-art upper bound (i.e., hardness) is $0.613$~\cite{huang2020fully,eckl2021stronger,tang2022improved}, separating the fully online model from the classic online bipartite matching model.

This setting is also known as the windowed online matching problem. \cite{ashlagi2023windowed} assumed a first-in-first-out structure on the active windows (i.e., arrivals and deadlines) of vertices, and achieved a $\nicefrac{1}{4}$ competitive ratio for edge-weighted graphs through a reduction to the Display Ads problem by suffering an extra factor of $2$. Combined with the state-of-the-art algorithm for the Display Ads problem by \cite{blanc2021multiway}, their competitive ratio can be improved to $0.268$.

\subsection{General Vertex Arrival}\label{sec:general}

Generalizing fully online matching is the online matching with general vertex arrivals problem, introduced by \cite{wang2015two}.
Again, the input is a graph $G=(V,E)$, initially unknown, with vertices arriving online. Upon the arrival of a vertex $v$, its incident edges to its previously-arrived neighbors are revealed. The algorithm either matches $v$ to an unmatched neighbor immediately or leaves $v$ unmatched, possibly matching it to a later-arriving neighbor $u$ upon $u$'s arrival. 
The inability to match vertices at any point before their departure (and lack of this information) makes this model more restrictive than the fully online model, and so algorithms for general vertex arrivals are also algorithms in the fully online model, with the same competitive ratio.

\cite{wang2015two} presented a fractional $0.526$-competitive algorithm for the fractional version of the problem. 
\cite{gamlath2019online} designed a rounding of Wang and Wong's fractional algorithm and established a $\nicefrac{1}{2}+\Omega(1)$ competitive ratio for the integral matching problem. This result stands as the only non-trivial integral algorithm so far. 
On the negative side, \cite{wang2015two,buchbinder2019edge,tang2022fractional} established an upper bound of $0.583$, separating the general vertex arrival model from the fully online model. 

\subsection{Edge Arrivals}

Finally, we remark that the most general online matching setting is the edge arrival model. That is, edges of an underlying graph arrive in a sequence and the algorithm decides whether to select an edge immediately on its arrival. 
Here edges correspond to fleeting collaboration opportunities between agents.
A competitive ratio of $\nicefrac12$ can be trivially achieved by a greedy algorithm that matches each edge on arrival if both its endpoints are free, and this is optimal for deterministic algorithms.
Unfortunately, \cite{gamlath2019online} proved that no online algorithm achieves a better than $\nicefrac12+\nicefrac{1}{2n}$ competitive ratio, even for the fractional version of the problem. 
Positive results are known assuming structure, including low-degree graphs and trees \cite{buchbinder2019edge}, batching \cite{lee2017maximum}, random-order arrivals \cite{guruganesh2017online}, or stochastic arrivals \cite{gravin2019online} (see \S\ref{sec:prophets} for more on the latter models).

\subsection{Reusable Resources}

In sponsored search, the advertisers’ budgets, viewed as resources, are non-reusable. In contrast, in such markets as cloud computing (e.g., AWS, Azure), short-term rentals (e.g., Airbnb), and freelancer labor (e.g., TaskRabbit), the allocated resources (be it compute, housing or labor) are reusable, and can be reallocated after being used.

The above motivates online bipartite matching  with reusable resources, where after an offline vertex (a rental service) is matched, it becomes available again after $d$ time steps, where $d$ is a known parameter that corresponds to the usage duration of the vertices. 
The classic online bipartite matching problem is a special case of the reusable resources model when $d = \infty$.

This model was first introduced by \cite{gong2022reusable} in a more general setting of online assortment optimization. \cite{goyal2021reusable,feng2021reusable} generalized the \balance algorithm and achieves an optimal $1-1/e$ competitive ratio for the fractional version of the problem.\footnote{Equivalently, they assumed that resources (offline vertices) have large capacities.}
For the integral version of the problem, \cite{delong2023reusable} proposed the \reranking algorithm (a variant of \ranking that reranks the offline vertices every $d$ time steps), and show that it achieves a competitive ratio of $0.589$, and an online correlated selection-based algorithm achieves a competitive ratio of $0.505$. All these results extend to the vertex-weighted setting. We remark that the results of \cite{delong2023reusable} heavily rely on the assumption that all vertices have identical usage durations $d$. The case of heterogeneous usage durations (i.e., each vertex $v$ has an individual duration time $d_v$) remains open.

\cite{feng2022reusable} further studied online assortment of reusable resources in the stochastic setting.
Reusable resources due to additional production have also been considered in infinite-horizon \emph{stochastic} settings \cite{aouad2020dynamic,collina2020dynamic,kessel2022stationary,patel2024combinatorial}. We discuss stochastic settings (without reusable resources) in more detail in the following section.

\section{Stochastic Models: Secretaries, Prophets, and Philosophers}\label{sec:prophets}

The preceding sections focused on adversarial models, where both input graph and arrival order are chosen by an adversary. 
This modeling choice, while robust, is quite pessimistic, and naturally results in worse guarantees than possibly achievable for real-world applications of interest.
A natural way to obtain improved provable guarantees is to either consider random arrival orders (but adversarial input), or to posit a stochastic generative model, possibly learnt from historical data.
Such models hearken back to classic results in optimal stopping theory concerning online Bayesian selection problems.

In the most basic setting, a buyer has a single item to sell, and impatient buyers arrive one after another and make take-it-or-leave-it bids for this single item. 
The buyer must select which bid to accept, immediately and irrevocably when the bid is made.
Under adversarial models, a buyer cannot be competitive with the hindsight-optimal solution.
In contrast, if the bids arrive in random order (referred to as the \emph{secretary problem}), then a competitive ratio of $\nicefrac{1}{e}$ is optimal \cite{dynkin1963optimum}. Similarly, if the successive bids $v_i$ are drawn independently
from known distributions $\calD_i$, then the optimal competitive ratio is $\nicefrac{1}{2}$ \cite{krengel1978semiamarts}, i.e., the buyer can guarantee an expected gain at least half of that obtained by a ``prophet'' who knows the realization of the randomness, 
$$\E[\textrm{Gain}]\geq \frac{1}{2}\cdot \E[\max_i v_i].$$
Such guarantees contrasting with the offline optimal, or prophet, are referred to as \emph{prophet inequalities}.
One may also contrast with the (computationally-unbounded) optimal online algorithm for such problems, which for reasons elaborated below we refer to as \emph{philosopher inequalities}.

These models have been generalized and extended significantly over the years.
In this section, we focus on recent developments for generalizations of the above to bipartite matching markets, where the buyer wishes to sell multiple heterogeneous items, and each arriving buyer proposes a different bid for each item. 
Put otherwise, we focus on online bipartite matching models.
We note that the buyer and seller terminology are not accidental, and these models have tight connections to questions in \emph{mechanism design}, which we also discuss in this section.

\begin{table}[ht]
	\centering
	\begin{tabular}{lcc}
		\toprule
		& impossibility & algorithmic \\
		\midrule
		Secretary  Matching & $\nicefrac{1}{e}$~\cite{dynkin1963optimum} & $\nicefrac{1}{e}$~\cite{kesselheim2013optimal} \\	
		{Prophet  Matching} & {$\nicefrac{1}{2}$~\cite{krengel1978semiamarts}} & $\nicefrac{1}{2}$~\cite{feldman2015combinatorial}\\
		Philosopher Matching & $0.99999$~\cite{papadimitriou2021online} & $0.652$~\cite{naor2023online} \\
		\bottomrule
	\end{tabular}	
	\caption{State-of-the-art for online bipartite weighted matching in stochastic settings.}
	\label{table:beyond-ads}
\end{table}

\subsection{Secretary Problems}

For edge-weighted online bipartite matching with online vertices arriving in random order, 
\cite{korula2009algorithms} were the first to obtain a constant-competitive ratio, specifically a $\nicefrac{1}{8}$-competitive algorithm. 
This was later improved to the optimal $\nicefrac{1}{e}$ ratio by \cite{kesselheim2013optimal}, generalizing the classic single-item result of \cite{dynkin1963optimum}, which we recall is the special case of a single offline vertex. 
For $k$ heterogeneous offline vertices, \cite{kleinberg2005multiple} showed  $(1-\nicefrac{1}{\sqrt{k}})$-competitive and truthful mechanism.
\cite{ezra2022general} study secretary matching in general graphs (with vertices arriving with edges to their previously-arrived neighbors, as in \S\ref{sec:fully} and \S\ref{sec:general}). They show that the optimal competitive ratio in this level of generality is $\nicefrac{5}{12}$, notably strictly greater than achievable for bipartite graphs.\footnote{For some intuition as to why this is not a contradiction, note that in the star example (i.e., the single-item problem), the center of the star arrives after the highest-bidding neighbor with probability $\nicefrac{1}{2}>\nicefrac{1}{e}$, and so greedily matching the center when it arrives is $\nicefrac{1
	}{2}$-competitive, and the lower bound of \cite{dynkin1963optimum} for bipartite graphs does not apply.}

The random-order model similarly allows for improved guarantees for the special cases of vertex-weighted and unweighted online bipartite matching.
For unweighted matching, \cite{guruganesh2017online} show that for \emph{edges} revealed in random order a better than $\nicefrac12$ competitive ratio is possible, notably beyond the worst-case optimal for adversarial arrivals \cite{gamlath2019online}.
Similarly, random-order vertex arrivals in bipartite graphs (with arrivals on only one side of the graph) allow one
 to surpass the worst-case optimal $1-\nicefrac1e$: a generalization of \textsc{ranking} achieves a competitive ratio of  $0.662$ \cite{huang2019online,jin2021improved}. 
This generalizes results of \cite{karande2011online,mahdian2011online}, who showed that for the unweighted problem \textsc{ranking} (unchanged) achieves competitiveness beyond $1-\nicefrac{1}{e}$, with the best known bound being $0.696$ \cite{mahdian2011online}.
As noted by these last two works, these results for random-order arrivals imply the same competitive ratios for stochastic matching problems with unknown i.i.d.~distributions over arrival types. This remains the best known result for unknown distributions.
In the following sections, we discuss the types of guarantees achievable under \emph{known} distributions.

\subsection{Prophet Inequalities}

The optimal competitive ratio of $\nicefrac12$ for the single-item problem due to \cite{krengel1978semiamarts} was also obtained several years later using a single-threshold (i.e., posted-price) algorithm by \cite{samuel1984comparison}. This results in truthful welfare-approximating mechanisms for single-item auctions. 
This connection between (pricing-based) prophet inequalities and mechanism design was later elaborated upon by researchers at the intersection of Economics and Computation \cite{hajiaghayi2007automated,chawla2010multi,kleinberg2019matroid}.
Interestingly, very recently
\cite{banihashem2024power} show that \emph{any} guarantee achieved by an online Bayesian selection algorithm can be achieved by a (dynamic) posted-price policy, implying that studying the non-strategic setting results in truthful mechanisms which achieve the same approximation of the social welfare as the algorithm in strategic settings.

The connection between (combinatorial) prophet inequalities and mechanisms design continues to motivate a flurry of results on prophet inequalities for increasingly involved markets, with more and more sophisticated combinatorial constraints on the sets of buyers that may be serviced, or items sold.
See the excellent surveys  \cite{hill1992survey,correa2019recent} and \cite{hartline2012approximation,lucier2017economic}for more on prophet inequalities and their connection to mechanism design, respectively. 
In what follows, we focus on prophet inequalities subject to matching constraints.

For unweighted online bipartite matching, \cite{feldman2009online} were the first to show that stochastic inputs allow for competitive ratio beyond the worst-case optimal $1-\nicefrac{1}{e}\approx 0.632$. Specifically, they show that if online vertices' neighborhoods are drawn i.i.d.~from a single known distribution, then a competitive ratio of $0.67$ is achievable. There has been a long line of work studying this question, most recently \cite{jaillet2013online,brubach2021improved,huang2021online,huang2022power,tang2022fractional}, with the current best competitive ratios being $0.7299$ and $0.716$ assuming integral and arbitrary arrival rates \cite{brubach2021improved,huang2022power}.\footnote{The arrival rate is the expected number of arrivals of a particular online type.}
For edge-weighted matching, a number of results were obtained under integral arrival rates \cite{haeupler2011online,brubach2021improved}, with the best ratio standing at $0.704$, while for arbitrary arrivals the ratio of $1-\nicefrac1e$ was only recently beaten \cite{yan2024edge,qiu2023improved}.
In contrast, by a work of  \cite{manshadi2012online}, no competitive ratios greater than $1-\nicefrac{1}{e^2}\approx 0.864$ and $0.823$ are possible in the same settings.

For unweighted and vertex-weighted bipartite matching under (much more general) time-varying independent distributions, \cite{tang2022fractional} recently provided the first algorithm surpassing the competitive ratio of $1-\nicefrac{1}{e}$, presenting a $0.666$-competitive algorithm.
For edge-weighted matching
a competitive ratio of $\nicefrac{1}{2}$ is best possible, as this generalizes the single-item problem of \cite{krengel1978semiamarts}. This ratio is known to be achievable via numerous approaches \cite{feldman2015combinatorial,dutting2020prophet,ezra2022prophet}.
This ratio of $\nicefrac{1}{2}$ is even achievable under vertex arrivals in general graphs \cite{ezra2022prophet}, or correlated arrivals~\cite{aouad2023nonparametric}. In contrast, the problem is strictly harder under edge arrivals, where the best known competitive ratio is in the range $[0.344,0.4]$ \cite{macrury2023random} and $[0.349,3/7]$ for bipartite graphs \cite{macrury2023random,correa2023optimal}. For unweighted matching a competitive ratio of $0.502$ is possible \cite{gravin2019online}.

\subsection{Philosopher Inequalities}

While a competitive ratio of $\nicefrac{1}{2}$ is worst-case optimal for online Bayesian selection subject to bipartite matching constraints, this is still a pessimistic worst-case guarantee, as the lower bounds focus on worst-case distributions. 
The optimal algorithms for distributions of interest may allow for better competitive ratios. 
This optimal algorithm, which is the solution of a Markov Decision Problem (MDP), is computable in polynomial space via standard techniques. 
As shown by \cite{papadimitriou2021online}, this is the right characterization, and even \emph{approximating} the optimal policy beyond some $0.999$ ratio is \textsc{pspace}-complete (i.e., is as hard as the hardest problems requiring polynomial space). Hence, under standard complexity-theoretic assumptions, this optimal policy is not computable in polynomial time. 
Put otherwise, it is likely computable only by a character with sufficient time to ``think'' (i.e., compute), therefore naturally referred to as a ``philosopher''. 
This motivates the study of polynomial-time approximation of the optimal online algorithm, which, in analogy with prophet inequalities (approximation of the optimal offline algorithm), we term \emph{philosopher inequalities}.

\cite{anari2019nearly} were the first to consider the approximation of the optimal online algorithm for online Bayesian selection. They considered bounded-depth and production-constrained laminar matroids, for which they provided $(1+\epsilon)$-approximate philosopher inequalities for any constant $\epsilon>0$.
\cite{dutting2023prophet} obtained the same bounds for random-order (secretary) philosopher inequalities for a single item.
For online bipartite matching, which by \cite{papadimitriou2021online} such an approximation would result in surprising developments in complexity theory. On the positive side, a successive line of work
\cite{papadimitriou2021online,saberi2021greedy,braverman2022max,naor2023online} showed that (increasingly) better than $\nicefrac{1}{2}$-approximate philosopher inequalities are possible. The current best bound stands at $0.652$ (notably, above the natural bound of $1-\nicefrac{1}{e}$ for online matching algorithms).
We note that the recent work of \cite{banihashem2024power} also translates the above (polynomial-time) policies approximating the optimal policy into pricing-based (and hence truthful) mechanisms providing the same approximation of the optimal mechanism.

\section{Overarching Techniques}\label{sec:techniques}

\subsection{Primal-Dual Algorithms}\label{sec:primal-dual}

The primal-dual method has found wide applications in the area of online algorithms. Refer to \cite{buchbinder2009design} for a comprehensive survey. For online matching and related problems, the primal-dual schema was first adapted by \cite{buchbinder2007online} to analyze \balance for the AdWords problem. We illustrate this idea in the special case of (unweighted) online bipartite matching.

We start with an economic interpretation of \balance. Consider the offline vertices as divisible items, and the online vertices as 0-1 unit-demand buyers. At any moment, each offline vertex $v$ \emph{prices} itself at $g(x_v)$ per (fractional) unit based on the current water level $x_v$ (matched fraction of $v$) and thus its neighbor receives a \emph{utility} of $1-g(x_v)$ per unit of $v$ (fractionally) assigned to it, where $g(\cdot)$ is an increasing function. (Note that the price of $v$ increases over time.) Upon its arrival, an online vertex $u$ continuously chooses the unmatched neighbors giving $u$ the largest utility.
Recall that the dual of the maximum matching problem is the minimum vertex cover problem. 
The economic interpretation suggests a natural way to set the dual variables: for each offline vertex, let its dual variable be the total collected price, and for each online vertex, let its dual variable be the utility.
The primal-dual framework asserts that in order to establish a $\Gamma$ competitive ratio, it suffices to prove that the total gain of each item-buyer pair is $\Gamma$. To illustrate this approach, we provide a formal yet brief analysis in Appendix~\ref{sec:lesson}.

Almost all fractional online matching algorithms were analyzed within the primal-dual framework. This includes algorithms for AdWords~\cite{buchbinder2007online}, Display Ads~\cite{devanur2016whole}, fully online matching~\cite{huang2019tight,huang2020fully2,tang2022improved}, general vertex arrivals~\cite{wang2015two}, stochastic matching~\cite{tang2022fractional}, etc.

The primal-dual method for fractional matching crucially requires the dual constraints to be satisfied \emph{always}. 
In contrast, \cite{devanur2013randomized} noticed that it suffices to have the dual constraints hold \emph{in expectation} for randomized algorithms, and used this observation to provide a simplified competitive analysis of \ranking for online (vertex-weighted) bipartite matching. Their approach is now referred to as the \emph{randomized} primal-dual schema.

Their proof relied on an intuitive economic interpretation of \ranking that is similar to the economic interpretation of \balance. Instead of maintaining a dynamic price that depends on the water level, each offline vertex sets a randomized fixed price (according to the random permutation generated by \ranking) at the beginning. Then on the arrival of each online vertex, it buys the cheapest remaining neighbor. Again, we split the gain of each matched edge between its two endpoints (i.e., set the corresponding dual variables), according to the price of the offline vertex and the utility of the online vertex. For our EC readers, refer to \cite{eden2021economics} for a proof that is written explicitly in the language of price and utility and avoids duality.

A remarkable property of the randomized primal-dual schema is its intrinsic robustness for vertex-weighted graphs for all variants of online bipartite matching. Indeed, this schema often (if not always) provides a ``free lunch'', allowing one
 to extend a result on unweighted graphs to vertex-weighted graphs while preserving the same competitive ratio. E.g., \cite{devanur2013randomized,huang2020online,huang2021online,huang2022power,tang2022fractional}.
Going beyond the online bipartite matching model,  in the fully online matching model \cite{huang2020fully} further developed the randomized primal-dual schema, by introducing a novel charging mechanic that allows a vertex other than the two endpoints of a matched edge to share the gain. \cite{levin2020streaming} further found an application of the randomized primal dual framework for submodular maximization.

\subsection{Randomized Rounding and Contention Resolution Schemes}

The relax-and-round framework considers fractional relaxations as guides for randomized algorithms' probabilistic choices. 
This section discusses the prevalence of this approach for online matching problems.

For bipartite matching, the standard relaxation allows us to assign a fractional value $x_e\in [0,1]$ to each edge so that any node $v$ has at most one unit assigned to its edges, $\sum_{e\ni v} x_e\leq 1$. 
Intuitively, $x_e$ can be thought of as the marginal matching probability of edge $e$ by some randomized algorithm, and use these fractions to obtain randomized algorithms.
Indeed, since every fractional bipartite matching is the convex combination of integral matchings $M_1,\dots,M_k$, this intuition can be made formal, by randomly picking one such matching $M_i$ with probability equal to its coefficient in the convex combination.
This results in each edge $e$ being matched with probability $x_e$, and thus preserves any linear objectives, $\sum_e w_e\cdot x_e$.
We refer to such rounding schemes matching each edge with probability $x_e$ as \emph{lossless} rounding schemes.

Perhaps surprisingly, and as hinted at by \Cref{table:central}, for \emph{online} edge-weighted bipartite matching and ad allocation problems, there exists a gap between our understanding of fractional algorithms and indivisible randomized algorithms.
Indeed, as pointed out by \cite[Footnote 3]{devanur2013randomized}, the above-mentioned integrality does not carry over to the online setting: 
for every randomized algorithm, there exist graphs on which the (optimally) $(1-\nicefrac1e)$-competitive fractional algorithm \balance achieves value $8/7$ times higher than any randomized algorithms. Therefore, rounding fractional algorithms seems to require losing a large multiplicative factor (in the worst case).
At face value, this large gap seems to rule out the use of the relax-and-round approach to obtain good randomized (integral) online matching algorithms.

Despite the above, a large number of results in online matching in recent years are obtained by (or can be interpreted as employing) online randomized rounding of fractional solutions, often obtained using the primal-dual schema, \S\ref{sec:primal-dual}. See \href{https://sites.google.com/view/focs23workshop-online-rounding/}{the FOCS23 workshop on the topic}. There are three flavors of results in this vein.

\subsubsection{Lossless rounding}
The Online Correlated Selection (OCS) technique, elaborated upon in \S\ref{sec:ocs} \cite{fahrbach2020edge,gao2021improved,blanc2021multiway} can be seen as losslessly rounding \emph{structured} fractional bipartite matchings in online settings.
\cite{buchbinder2023lossless} were the first to explicitly ask what structure is \emph{necessary} for lossless online rounding of bipartite matching, i.e., allowing one to match each edge 
 with probability \emph{exactly} $x_e$.
They considered other structured fractional online matching algorithms and provided lossless online rounding schemes for these, which they used to obtain generalizations of OCS and sharp randomness thresholds for beating deterministic algorithms for online bipartite matching.
Similarly, ``spread out'' fractional matchings, e.g., ones assigning value $1/\Delta$ in graphs of maximum degree $\Delta$, can be rounded \emph{nearly} losslessly, i.e., one can match each edge $e$ with probability $x_e\cdot (1-\epsilon)$, and this is key to numerous results for online edge coloring, e.g., \cite{cohen2018randomized,wajc2020matching,blikstad2024online}.

\subsubsection{Approximate Rounding}
In the other extreme, \cite{naor2023online} ask how well \emph{arbitrary} fractional matchings $\vec{x}$ can be rounded online, and provide approximate rounding schemes that match each edge $e$ with probability $0.652\cdot x_e$, notably breaking the barrier of $1-\nicefrac{1}{e}$ for this problem. They then use this scheme to obtain improved results for online edge coloring of multigraphs and philosopher inequalities, among others.
Some prior results for philosopher inequalities \cite{papadimitriou2021online,saberi2021greedy} are also obtained by such approximate online rounding schemes applied to LP relaxations incorporating constraints only applicable to online algorithms \cite{torrico2022dynamic,buchbinder2014secretary}.
Similarly, the multiway OCS of \cite{gao2021improved} can be seen (and used) as an approximate rounding scheme that provides guarantees per offline vertex, as opposed to per edge, lending itself to results for vertex-weighted matching.
The result of \cite{gamlath2019online} for online matching under general vertex arrivals likewise follows a lossy rounding approach (applied to the fractional algorithm of \cite{wang2015two}), though here the approximation guarantees are more global than per-vertex or per-edge.
More approximate rounding schemes are obtained by Online Contention Resolution Schemes (OCRS), whose guarantees are weak in the context of adversarial settings, but are central to prophet inequalities, as we now discuss.

\subsection{Online Contention Resolution Schemes}

Contention resolution schemes (CRS) have their origins in the (offline) submodular optimization literature \cite{chekuri2014submodular}, and follow a natural rounding approach: 
\emph{activate} each element independently with probability $x_e$, and then select a high-valued feasible subset (in our case, a matching) among the active elements. 
This is obtained by guaranteeing each active element be selected with as high a probability as possible. This probability $\Pr[e \textrm{ selected} \mid e \textrm{ active}]$ is referred to as the \emph{balance ratio} of the CRS.

The above approach can be generalized to online settings \cite{feldman2016online}, where inclusion of active element (in our case, edge) $e$ must be made immediately upon its activation.
By an approach due to \cite{yan2011mechanism}, using an appropriate convex \emph{ex-ante relaxation}, an OCRS with balance ratio $c$ provides $c$-approximate prophet inequalities for the same setting, by considering an element (buyer) active if their bid is in their $x_e$-th percentile.
As \cite{lee2018optimal} show, the opposite is true: prophet inequalities that are $c$-competitive with respect to this relaxation yield OCRS that are $c$-balanced.
By preceding discussions on pricing-based prophet inequalities being derivable from arbitrary prophet inequalities, we find that OCRS yield welfare-approximating truthful mechanisms \cite{banihashem2024power}.
Generally, OCRS have found widespread applications since their introduction. See \cite{patel2024combinatorial} for a discussion.

When specified to matchings, for batched OCRS (capturing vertex arrivals), balance ratios of $\nicefrac{1}{2}$ and $(1 + \nicefrac{1}{e^2})/2 \approx 0.567$ are optimal for adversarial vertex arrivals in general graphs  \cite{ezra2022prophet} and random-order vertex arrivals in bipartite graphs \cite{macrury2024random}, respectively. 
For other settings, despite much progress in recent years, the optimal balance ratio is still unknown for adversarial edge arrivals \cite{gravin2019prophet,correa2023optimal,macrury2023random}, random-order edge arrivals \cite{brubach2021improved,pollner2022improved,macrury2023random} and random-order vertex arrivals \cite{fu2021random,macrury2024random}.

\subsection{Online Rounding: Online Correlated Selection}
\label{sec:ocs}

In this section we elaborate on one particular general rounding scheme for online bipartite matching and its generalizations, termed 
Online Correlated Selection (OCS). 
In each round, the OCS observes the online item \emph{and a fractional allocation of this item} to the offline agents.
It must then allocate this item in whole to one of the agents.
We measure the quality of an OCS by how an agent's value for the allocated subset of items depends on the fractional allocation to the agent.

Consider unweighted matching as a running example.
We expect that an offline agent gets matched with a higher chance as the total fractional allocation to the agent increases.
Consider  the following natural baseline algorithm:
match each item $i$ to an unmatched agent $a$ interested in the item with probability proportional to $x_{ai}$, the fractional allocation of item $i$ to agent $a$.
Let $x_a$ denote the total fractional allocation to agent $a$.
This rounding algorithm guarantees that agent $a$ is matched with probability at least $1 - e^{-x_a}$, and this bound is tight (c.f., \cite{gao2021improved}).
Unfortunately, combining the Balance algorithm and this baseline bound only gives the trivial $0.5$ competitive ratio.
See \cite{gao2021improved} for the state-of-the-art techniques for designing OCS for unweighted and vertex-weighted matching, which improved the above bound for the probability of matching agent $a$, and the resulting competitive ratios for these two problems.
\cite{hosseini2023class} further applied the OCS technique to an unweighted online matching problem with class fairness as its main objective.

The definitions of OCS for display ads and AdWords are too technical to be covered concisely in this brief survey.
We refer readers to \cite{fahrbach2020edge} for the original definition of OCS for display ads and a proof of concept that non-trivial OCS exists, and \cite{blanc2021multiway} for the best existing result along this line.
See \cite{huang2020adwords} for the definition of OCS for AdWords;
the technique is similar to the proof of concept for display ads by \cite{fahrbach2020edge}.

Although the concept of OCS was first introduced for problems in non-stochastic models, the technique has found applications in stochastic models as well.
In hindsight, this is perhaps unsurprising. 
The general recipe for online matching algorithms in stochastic models is to first solve an LP relaxation and then make online matching decisions taking the LP solution as a guide.
By design, an OCS can treat the LP solution as a fractional allocation and convert it into online matching decisions.
\cite{tang2022fractional} used this approach to obtain the first non-trivial algorithm for the non-IID model of unweighted and vertex-weighted online stochastic matching.
\cite{huang2022power} further gave an OCS tailored for the stochastic model to get the best competitive ratios to date for IID unweighted and vertex-weighted online stochastic matching.

\subsection*{Acknowledgements}
Zhiyi Huang is supported by an NSFC grant 6212290003.
Zhihao Gavin Tang is supported by NSFC grant 61932002.
David Wajc is supported in part by a Taub Family Foundation ``Leader in Science and Technology'' fellowship.
We thank Niv Buchbinder for his short analysis of the \balance algorithm (in the appendix), and Sam Taggart for valuable feedback.

\appendix 

\section{ADDENDUM: A Teachable Moment}\label{sec:lesson}

Following a quote attributed to Feynman, namely ``If you want to master something, teach it,'' we present short and self-contained (and in our opinion, quite teachable) proofs of two basic results in online bipartite matching: a competitive ratio of $1-\nicefrac1e$ for \balance extended to vertex-weighted matching \cite{buchbinder2007online}, and an even shorter proof that no fractional algorithm can do better.

Recall that for the online vertex-weighted bipartite matching problem, each offline vertex $i$ and its positive weight $w_i\geq 0$ are known up front. At each time $t$, online vertex $t$ arrives, together with its edges $(i,t)\in E$ to its neighbors $i\in N(t)$, and we must decide to what extent $x_{i,t}$ to assign $t$ to its neighbor $i$, from which the algorithm accrues a value of $w_i\cdot x_{i,t}$. 
Both offline vertices and online vertices must be assigned to a total extent of at most one.

A linear programming (LP) relaxation of the problem (allowing us to match each edge $(i,t)$ to an extent of $x_{i,t}$), together with this LP's dual, are as follows.

\begin{align*}
    \textrm{(P)\,\,} \max\,\, & \sum_{(i,t)\in E} w_i\cdot x_{i,t}  & \textrm{(D)\,\,} \min\,\, & \sum_i y_i + \sum_t z_t \\
    \textrm{s.t.\,\,} & \sum_t x_{i,t} \leq 1 \quad \forall i & \textrm{s.t.\,\,} & y_i + z_t \geq w_i \quad \forall (i,t)\in E \\
        & \sum_i x_{i,t} \leq 1 \quad \forall t & &  y,z\geq 0  \\
        & x \geq 0 & & 
\end{align*}

\subsection{Short analysis of \balance}

\underline{The \balance algorithm:}
Initialize zero primal and dual solutions, $\vec{x}$ and $\vec{y},\vec{z}$. Let $g(x):=\frac{e^x-1}{e-1}$. 
For every online vertex $t$ on arrival, letting $x_i := \sum_{t'<t}x_{i,t'}$ be the fractional degree of neighboring vertex $i\in N(t)$ before arrival of $t$, we increase $x_{i,t}$ for all $i\in A,$ where
$$A:=\underset{i\in N(t)}{\arg\max}\{w_i\cdot (1-g(x_i+x_{i,t}))\},$$
so that this set $A$ grows monotonically, 
until $\sum_i x_{i,t}=1$ or $\min_{i\in N(t)} \big(x_i + x_{i,t}\big)=1$. 

Note that for the unweighted case ($w_i=1$ for all $i$), this algorithm is precisely the \wl algorithm described in \S\ref{sec:central}.
The intuition behind this algorithm is clear: since each offline vertex is equally likely to have no future neighbors, we wish to maximize the value assigned to edges of the least (fractionally) matched offline vertex, in case it (and it alone) has no future edges. That generalizing this approach allows to get a competitive ratio better than $1/2$ (also for vertex-weighted matching) is perhaps less immediate. We present a proof of this fact using LP duality, and specifically dual fitting.

\underline{Dual fitting:} 
For our analysis (only), for each offline vertex $i$ and online vertex $t$, we set dual values 
\begin{align*}
 y_i & \gets w_i\cdot g\left(\sum_{t'} x_{i,t'}\right),\\
 z_t & \gets \max_{i\in N(t)} \left(w_i \cdot \left(1-g\left(\sum_{t'\leq t} x_{i,t'}\right)\right)\right).
\end{align*}

The first step of any primal-dual-based proof involves showing that the constructed dual is feasible, and hence its cost upper bounds the maximum gain (in hindsight).
\begin{lemma}\label{WF:dual-feasible}
    Vectors $\vec{y},\vec{z}$ are dual feasible, i.e., they are positive and $y_i + z_t \geq w_i$ for all edges $(i,t)\in E$.
    Consequently, $\sum_i y_i + \sum_t z_t \geq OPT$, where $OPT$ is the weight of a maximum vertex-weighted matching.
\end{lemma}
\begin{proof}
As $g:[0,1]\to [0,1]$, we have that $y_i,z_t\geq 0$.
On the other hand, by monotonicity of $g$ (and positivity of $\vec{x}$), for every edge $(i,t)\in E$ we have that $y_i = w_i\cdot g(\sum_{t'}x_{i,t'}) \geq w_i\cdot g(\sum_{t'\leq t}x_{i,t'})$, and so  
\begin{align*}
y_i + z_t & \geq w_i\cdot g\left(\sum_{t'\leq t}x_{i,t'}\right) + w_i\cdot \left(1-g\left(\sum_{t'\leq t}x_{i,t'}\right)\right) = w_i.
\end{align*}
The lower bound $\sum_i y_i + \sum_t z_t \geq OPT$ then follows from weak LP duality, together with the primal LP being a fractional relaxation of the problem.
\end{proof}

The second step in a primal-dual proof involves bounding the ratio of the primal and dual solutions' gain/cost. For online algorithms this typically boils down to bounding the ratio of the change in these values per cost, as in the following.
\begin{lemma}\label{WF:PD}
For each time $t$, the increase in primal value, $\Delta P_t := \sum_i w_i\cdot x_{i,t}$, and the increase in the dual cost, $\Delta D_t := z_t + \sum_i \left(w_i\cdot g(x_i+x_{i,t}) - w_i\cdot g(x_i)\right)$, satisfy $$\Delta P_t /\Delta D_t \geq 1-\nicefrac{1}{e}.$$
\end{lemma}
\begin{proof}
The increase in dual cost satisfies
\begin{align*}
    \Delta D_t & = \sum_i w_i\cdot \left(g(x_i + x_{i,t}) - g(x_i)\right) + z_t \\
    & \leq \sum_i w_i\cdot x_{i,t} \cdot \left(g(x_i+x_{i,t}) + \frac{1}{e-1}\right) + z_t \\ 
    & = \sum_i w_i\cdot x_{i,t} \cdot \left(g(x_i+x_{i,t}) + \frac{1}{e-1}\right) + \max_i w_i\cdot (1-g(x_i+x_{i,t})) \\ 
    & = \sum_i w_i\cdot x_{i,t} \cdot \left(g(x_i+x_{i,t}) + \frac{1}{e-1}\right) + \sum_i x_{i,t} \cdot \max_j w_j\cdot (1-g(x_j+x_{j,t})) \\ 
    & = \sum_i w_i\cdot x_{i,t} \cdot \left(g(x_i+x_{i,t}) + \frac{1}{e-1}\right) + \sum_i x_{i,t} \cdot w_i\cdot (1-g(x_i+x_{i,t})) \\ 
    & = \sum_i w_i\cdot x_{i,t} \cdot \left(1 + \frac{1}{e-1}\right) \\
    & = \Delta P_t \cdot \frac{e}{e-1}.
\end{align*}
Above, the single inequality relied on the definition of $g(x) = \frac{e^x-1}{e-1}$ implying that $g'(x) = g(x) + \frac{1}{e-1}$ is monotone increasing in $x$. The second and third equalities follow by definition of $z_t$ and either $\sum_i x_{i,t}=1$ or $\min_{i\in N(t)} \big(x_i + x_{i,t}\big)=1$, where the latter implies that $g(x_i+x_{i,t})=1$ for all $i\in N(t)$. The fourth equality follows from $x_{i,t}=0$ unless $i\in \underset{i\in N(t)}{\arg\max}\{w_i\cdot (1-g(x_i+x_{i,t}))\}.$
\end{proof}

Combining both lemmas, the algorithm's competitive ratio follows.
\begin{theorem}
Algorithm \balance is $(1-\nicefrac{1}{e})$-competitive for vertex-weighted online bipartite matching.
\end{theorem}
\begin{proof}
    Summing over all times $t$, we have by  Lemmas \ref{WF:PD} and \ref{WF:dual-feasible} that 
    \begin{align*}
        \sum_{i,t} w_i\cdot x_{i,t} = \sum_t \Delta P_t \geq (1-\nicefrac{1}{e})\cdot \sum_t \Delta D_t = \sum_i w_i \cdot g\left(\sum_{t'} x_{i,t'}\right) + \sum_t z_t \geq OPT,
    \end{align*}
    where the last equality uses the telescoping sum 
    \begin{align*}
        \sum_t \left( g\left(\sum_{t'\leq t}x_{i,t'}\right) - g\left(\sum_{t'< t}x_{i,t'}\right)\right) & = g\left(\sum_{t'}x_{i,t'}\right)-g(0)=g\left(\sum_{t'}x_{i,t'}\right),
    \end{align*}
    where the last equality follows from $g(0)=0$.
\end{proof}

\subsection{Matching Impossibility Result}

We now show that \balance is optimal, up to vanishingly small lower-order terms, even for the special case of \emph{unweighted} online bipartite matching.
The underlying input family is the same as that used by \cite{karp1990optimal} to prove that randomized algorithms are at best $(1-\nicefrac1e)$-competitive. The following proof shows in a more direct manner that this ratio is best possible even for (potentially more powerful) fractional algorithms.
The intuition behind the proof is precisely that guiding \wl described above.

\begin{theorem}\label{thm:adversarial_lb}
Any fractional online bipartite (unweighted) matching algorithm $\mathcal{A}$ has competitive ratio at most $1-\nicefrac{1}{e}+o(1)$. 
\end{theorem}
\begin{proof}
We consider inputs consisting of $n$ offline and online vertices. The first online vertex neighbors all offline vertices; each following online vertex has the same neighborhood as its predecessor, barring one vertex. (Under appropriate labeling the ``bipartite'' adjacency matrix of this graph is upper triangular.) This graph clearly has a perfect matching, and therefore optimum value of $n$.
To prove our theorem we show that for any fractional algorithm $\mathcal{A}$, a judicious choice of input forces $\mathcal{A}$ to achieve value at most $n(1-\nicefrac{1}{e})+1$.

For every online vertex $t$, the neighbor $i$ with no future edges is chosen to be a neighbor which was matched below the average, e.g., which minimizes $\sum_{t'\leq t} x_{i,t'}$. 
As the $t$-th online vertex to arrive neighbors $n-t+1$ offline vertices, a proof by induction shows that offline vertices that do neighbor online vertex $t+1$ are matched before time $t+1$ to an average of at least $\sum_{t'\leq t}x_{i,t'}\geq \min\{1,\sum_{t'=1}^{t}\frac{1}{n-t'+1}\}$.\footnote{This holds only if the algorithm $\mathcal{A}$ is greedy, in the sense that it exhausts every online vertex $t$, i.e., sets $\sum_i x_{i,t}=1$, when feasible. However, it is easy to modify any algorithm with a greedy algorithm which does no worse.} Consequently, for $t=n(1-\nicefrac{1}{e})+1$, every offline vertex $i$ that neighbors $t+1$ is fractionally matched to an extent of at least one, since
\[
\sum_{t'=1}^{n(1-\nicefrac{1}{e})+1}\frac{1}{n-t'+1}=\sum_{i=\nicefrac{n}{e}}^{n}\frac{1}{i}\geq \int_{\nicefrac{n}{e}}^{n+1}\frac{1}{x}\  dx=\ln(n+1)-\ln(\nicefrac{n}{e}) \geq 1.
\]
But since $\sum_{t'} x_{i,t'}\leq 1$, this implies that every offline vertex $i$ neighboring online vertices after time $t$ is already fully (fractionally) matched by time $t$, and so the algorithm gains no further profit from such $i$.
Consequently,  $\mathcal{A}$ achieves a value of at most $\sum_{i,t'} x_{i,t'}\leq n(1-\nicefrac{1}{e})+1$, and so has competitive ratio at most $1-\nicefrac{1}{e}+\nicefrac{1}{n}$.
\end{proof}

See \cite{feige2018tighter} for tight bounds on the $o(1)$ term in the above upper bound.

\bibliographystyle{alpha}
\bibliography{refs}

\end{document}